\documentclass[aps,pra,10pt,twocolumn,superscriptaddress]{revtex4-2}

\usepackage{amsfonts,amsmath,amssymb,amsthm,mathtools}
\usepackage{braket}
\usepackage{dcolumn}
\usepackage{graphicx}
\usepackage{tikz}
\usepackage{xcolor}
\usepackage{placeins}
\usepackage[colorlinks=true, linkcolor=blue, citecolor=blue, urlcolor=blue]{hyperref}
\IfFileExists{orcidlink.sty}{\usepackage{orcidlink}}{%
  \newcommand{\orcidlink}[1]{\href{https://orcid.org/##1}{\textsuperscript{ORCID}}}}
\usepackage[capitalize]{cleveref}

\crefname{equation}{Eq.}{Eqs.}
\Crefname{equation}{Eq.}{Eqs.}
\crefname{section}{Sec.}{Secs.}
\Crefname{section}{Sec.}{Secs.}
\crefname{subsection}{Sec.}{Secs.}
\Crefname{subsection}{Sec.}{Secs.}
\crefname{subsubsection}{Sec.}{Secs.}
\Crefname{subsubsection}{Sec.}{Secs.}
\crefname{figure}{Fig.}{Figs.}
\Crefname{figure}{Fig.}{Figs.}
\crefname{appendix}{Appendix}{Appendices}
\Crefname{appendix}{Appendix}{Appendices}
\crefname{table}{Table}{Tables}
\Crefname{table}{Table}{Tables}
\crefname{proposition}{Proposition}{Propositions}
\Crefname{proposition}{Proposition}{Propositions}
\crefname{theorem}{Theorem}{Theorems}
\Crefname{theorem}{Theorem}{Theorems}
\crefname{corollary}{Corollary}{Corollaries}
\Crefname{corollary}{Corollary}{Corollaries}

\usetikzlibrary{positioning,shapes.geometric}

\newtheorem{proposition}{Proposition}
\newtheorem{theorem}[proposition]{Theorem}
\newtheorem{corollary}[proposition]{Corollary}

\newcommand{\C}{\mathbf{C}}
\newcommand{\Np}{\mathbf{N}_{\geq 0}}

\definecolor{mblue}  {rgb}{0.368417, 0.506779, 0.709798}
\definecolor{morange}{rgb}{0.880722, 0.611041, 0.142051}
\definecolor{mgreen} {rgb}{0.560181, 0.691569, 0.194885}
\definecolor{mred}   {rgb}{0.922526, 0.385626, 0.209179}
\definecolor{mpurple}{rgb}{0.647624, 0.37816,  0.614037}
\definecolor{mcyan}  {rgb}{0.363898, 0.618501, 0.782349}

\tikzset{mps node/.style={circle, draw, inner sep=1pt, minimum size=1.1cm,fill=morange!75}}
\tikzset{copy/.style={circle,
            inner sep=1pt,
            minimum size=0.2cm,
            fill=black}}
\tikzset{
    left/.style={regular polygon, regular polygon sides=3,
    draw, fill=morange!75,
    minimum size = 1cm,
    inner sep = 0pt,
    shape border rotate=90,}}
\tikzset{
    right/.style={regular polygon, regular polygon sides=3,
    draw, fill=morange!75,
    minimum size = 1cm,
    inner sep = 0pt,
    shape border rotate=270,}
    }
\tikzset{
    mpo node/.style={
    circle,
    draw,
    inner sep=1pt,
    minimum size=1.1cm,
    fill=mblue!75}
    }
\tikzset{
    basis change/.style={
    draw,
    inner sep=1pt,
    minimum size=1cm,
    fill=mpurple!75}
    }

\begin{document}

\title{Basis-update and Galerkin time integration in canonical matrix-product-state form}
\author{Maximilian Fröhlich\,\orcidlink{0009-0007-5276-2858}}
\affiliation{Weierstrass Institute for Applied Analysis and Stochastics, Berlin, Germany}
\author{Richard M.~Milbradt \orcidlink{0000-0001-8630-9356}}
\affiliation{Technical University of Munich, Munich, Germany}
\author{Martin Eigel\,\orcidlink{0000-0003-2687-4497}}
\affiliation{Weierstrass Institute for Applied Analysis and Stochastics, Berlin, Germany}
\author{Aaron Sander\,\orcidlink{0009-0007-9166-6113}}
\affiliation{Technical University of Munich, Munich, Germany}
\author{Robert Wille\,\orcidlink{0000-0002-4993-7860}}
\affiliation{Technical University of Munich, Munich, Germany}
\affiliation{MQSC GmbH, Munich, Germany}
\affiliation{Software Competence Center Hagenberg GmbH (SCCH), Hagenberg, Austria}

\author{Christian B.~Mendl~\orcidlink{0000-0002-6386-0230}}
\affiliation{Technical University of Munich, Munich, Germany}

\begin{abstract}
Matrix product state algorithms must enlarge their bond spaces as entanglement grows and compress them to control cost. We formulate basis-update and Galerkin (BUG) time integration as a sequence of canonical MPS sweeps for Hamiltonians represented as matrix product operators. We show when two natural basis updates produce the same trial space and when transporting coefficients between successive bases preserves the represented state. Under these conditions, the existing first-order error bound for uncompressed tree-tensor-network BUG also applies to the alternating-endpoint MPS schedule. We verify the uncompressed implementation against an independent six-site calculation. We then compare BUG with two-site TDVP for 16-site transverse-field Ising and Haldane–Shastry dynamics. At matched timestep and truncation settings, BUG performs fewer local exponential actions and has lower runtime. These settings do not produce equal accuracy. The runtime versus accuracy curves cross for the Ising model and are close for the Haldane–Shastry model. The comparison therefore identifies model-dependent trade-offs rather than a general advantage for either method.
\end{abstract}


\maketitle

\section{Introduction}
Classical simulation helps characterize quantum many-body dynamics and assess emerging quantum devices.  The many-body Hilbert space grows exponentially with system size, and real-time evolution becomes harder as entanglement develops.  Matrix product states (MPS) reduce this cost when the state remains sufficiently compressible.  Their practical reach therefore depends on whether the virtual bond spaces can grow to represent new correlations and then be compressed without losing essential information.

Established time-evolution methods handle bond growth in different ways.  Time-evolving block decimation is particularly effective when the Hamiltonian has a local gate decomposition \cite{Vidal2003,Vidal2004}.  Methods based on the time-dependent variational principle (TDVP) apply directly to general matrix product operator (MPO) Hamiltonians \cite{Paeckel2019, Haegeman_2011}.  Rank-adaptive schemes such as controlled bond expansion enlarge the available MPS manifold as entanglement grows \cite{Li_2024}.  Each approach distributes time-discretization, projection, and compression errors differently.  A basis-update method must in addition enlarge the basis without losing directions already used to represent the state.

The basis-update and Galerkin (BUG) method was introduced for dynamical low-rank matrices \cite{Ceruti2022}, extended to Tucker tensors \cite{Ceruti2022B}, and later formulated for tree tensor networks (TTNs) \cite{Ceruti2023,Sulz2025}.  Parallel TTN variants have also been developed \cite{Ceruti2024ParallelBUGTTN}.  BUG enlarges the approximation space before applying a Galerkin update.  Its analysis relates the spaces before and after this enlargement.  Because tensor trains and MPS are chain-shaped TTNs, these results establish that BUG applies to MPS in principle.

MPS-based implementations have since made different choices within this general construction.  One keeps an orthogonality center near the middle of the chain, updates both sides toward it, and then evolves the center before compressing the state \cite{petersson2026dynamicalsimulationsschrodingersequation}.  Another gives a parallel MPS/MPO realization for hybrid tensor networks and explicitly maps between updated bases \cite{bauer2026timeevolutionhybridtensor}.  These formulations show that BUG can be organized in more than one way in canonical MPS form.

The abstract TTN formulation does not by itself fix every operation in a canonical MPS implementation.  The implementation must choose where to perform the final projected update, how to enlarge each bond basis, how to project the coefficients into the updated basis, and when to compress.  These choices determine whether the updated representation retains the previous basis directions and whether moving the coefficients between bases changes the represented state.

Here, we formulate a sequential canonical-MPS realization of BUG for time-independent Hermitian MPO Hamiltonians.  Each sweep builds updated bond bases toward a fixed endpoint and performs the final projected evolution there.  We compare enlarging the basis with the current center tensor against retaining the previous basis explicitly.  Standard rank-adaptive BUG includes the old basis precisely so that the projected starting tensor is unchanged \cite{Ceruti2022B,Ceruti2023}.  We make explicit, for the center-augmented MPS variant, the full-rank condition under which the two choices generate the same space before compression.  We also specialize the Pythagorean identity for orthogonal projection to coefficient transport between the previous and updated block bases, quantifying when the represented block is preserved and what is lost otherwise.  Under these inclusion conditions, we transfer the published robust first-order TTN-BUG error bound to the uncompressed alternating-endpoint MPS schedule.

An independent six-site calculation then checks the complete uncompressed sweep.  We compare the resulting rank-adaptive method with two-site TDVP (2-TDVP) for 16-site transverse-field Ising and periodic Haldane-Shastry dynamics.  To isolate the algorithms, both methods start from the same MPS and use identical MPO tensors, truncation parameters, bond caps, and Krylov solver.  BUG performs fewer local matrix exponentials and is faster at every matched setting tested.  The accuracy comparison is, however, more nuanced, in that 2-TDVP is more accurate for the Ising calculations, while the long-range result depends on the timestep and tolerance.  We therefore report runtime-infidelity trade-offs for all test cases rather than asserting a universal advantage.  

The remainder of this paper is organized as follows. \Cref{sec:mps_mpo} introduces the MPS and MPO notation used throughout the paper. \Cref{sec:BUG} presents the endpoint-rooted BUG sweep, compares the two basis-enlargement choices, and analyzes coefficient transport and compression. \Cref{sec:numerical_results} first checks the uncompressed implementation against independent dense dynamics and then compares rank-adaptive BUG with 2-TDVP for the transverse-field Ising and Haldane--Shastry models. \Cref{sec:discussion} summarizes the results, discusses their scope, and identifies directions for further study. An open-source implementation of the methods used in this work is available in Yet Another Quantum Simulator (YAQS), part of the Munich Quantum Toolkit (MQT) \cite{YAQS,wille_mqt2024}.
 
\section{Matrix product states and operators} \label{sec:mps_mpo}
This section establishes the MPS, canonical-form, and MPO conventions needed for the construction below.  \Cref{fig:mps_mpo} summarizes these representations and the isometric contractions used throughout the algorithm.

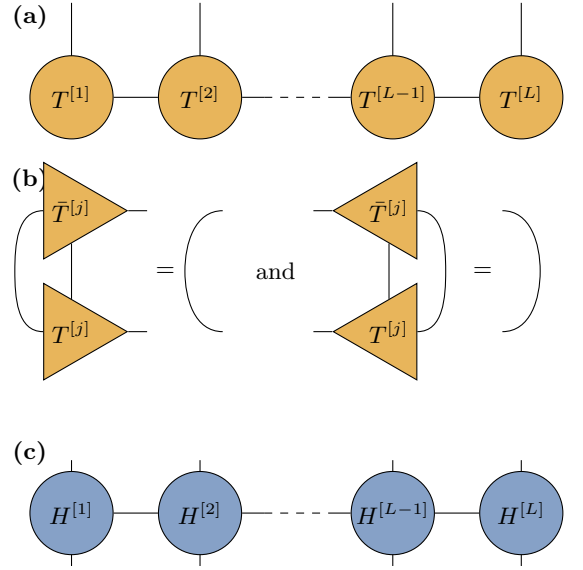
\begin{figure}[t]
    \centering
    \begin{tikzpicture}
        \def\separ{1.7}
        \def\physsep{1.25}
        \node[font=\bfseries] at (-0.55,0.85*\physsep) {(a)};
        \foreach \i in {0,1,2.5,3.5} {
            \draw (\i*\separ,0) -- (\i*\separ,\physsep);
        }
        \node[mps node] (T1) at (0,0) {$T^{[1]}$};
        \node[mps node] (T2) at (\separ,0) {$T^{[2]}$};
        \node[mps node] (TL1) at (2.5*\separ,0) {$T^{[L-1]}$};
        \node[mps node] (TL) at (3.5*\separ,0) {$T^{[L]}$};
        \draw (T1) -- (T2) -- (1.5*\separ,0);
        \draw[dashed] (1.5*\separ,0) -- (2*\separ,0);
        \draw (2*\separ,0) -- (TL1) -- (TL);

        \begin{scope}[shift={(0,-3.1)}]
            \def\horzsep{1}
            \def\vertsep{1.6}
            \node[font=\bfseries] at (-0.55,1.25*\vertsep) {(b)};
            \node[right] (CT) at (0,0) {$T^{[j]}$};
            \node[right] (CTs) at (0,\vertsep) {$\Bar{T}^{[j]}$};
            \draw (CT) to[out=180,in=-90] (-0.75*\horzsep,0.5*\vertsep) to[out=90,in=180] (CTs);
            \draw (\horzsep,0) -- (CT) -- (CTs) -- (\horzsep,\vertsep);
            \node[anchor=west] at (\horzsep,0.5*\vertsep) {=};
            \draw (2*\horzsep,\vertsep) to[out=180,in=90] (1.5*\horzsep,0.5*\vertsep) to[out=-90,in=180] (2*\horzsep,0);
            \node at (2.7,0.5*\vertsep){and};
            \begin{scope}[shift={(4.2,0)}]
                \node[left] (CT) at (0,0) {$T^{[j]}$};
                \node[left] (CTs) at (0,\vertsep) {$\Bar{T}^{[j]}$};
                \draw (CT) to[out=0,in=-90] (0.75*\horzsep,0.5*\vertsep) to[out=90,in=0] (CTs);
                \draw (-\horzsep,0) -- (CT) -- (CTs) -- (-\horzsep,\vertsep);
                \node[anchor=west] at (\horzsep,0.5*\vertsep) {=};
                \draw (1.5*\horzsep,\vertsep) to[out=0,in=90] (2*\horzsep,0.5*\vertsep) to[out=-90,in=0] (1.5*\horzsep,0);
            \end{scope}
        \end{scope}

        \begin{scope}[shift={(0,-5.5)}]
            \node[font=\bfseries] at (-0.55,0.8) {(c)};
            \foreach \i in {0,1,2.5,3.5} {
                \draw (\i*\separ,-0.7) -- (\i*\separ,0.7);
            }
            \node[mpo node] (H1) at (0,0) {$H^{[1]}$};
            \node[mpo node] (H2) at (\separ,0) {$H^{[2]}$};
            \node[mpo node] (HL1) at (2.5*\separ,0) {$H^{[L-1]}$};
            \node[mpo node] (HL) at (3.5*\separ,0) {$H^{[L]}$};
            \draw (H1) -- (H2) -- (1.5*\separ,0);
            \draw[dashed] (1.5*\separ,0) -- (2*\separ,0);
            \draw (2*\separ,0) -- (HL1) -- (HL);
        \end{scope}
    \end{tikzpicture}
    \caption{\textbf{Matrix product state representations.} These diagrams define the tensor conventions used in the canonical BUG construction.  \textbf{(a)} An MPS represents the state as a chain of site tensors.  \textbf{(b)} Left- and right-isometric contractions define canonical form and identify the orthogonality center.  \textbf{(c)} An MPO represents the Hamiltonian in the same chain geometry.  The canonical isometries make the block embeddings orthonormal and allow the environments and basis updates below to be written as local contractions.}
    \label{fig:mps_mpo}
\end{figure}

Broader accounts of MPS time evolution and tensor networks are available in Refs.~\cite{Paeckel2019,Schollwock2011,Or_s_2019}.  Matrix product states efficiently represent many-body states with limited entanglement \cite{Schollwock2011}.  For local dimension $d\in\Np$ and system size $L\in\Np$, any state $\ket\psi\in(\C^d)^{\otimes L}$ has an MPS representation
\begin{equation}\label{eq:mps_symb}
\begin{split}
    \ket{\psi} = \sum_{\substack{\{\nu_i\}_{i=1}^{L-1}\\
                                 \{p_i\}_{i=1}^{L}}}
                                    T^{[1]}_{p_1\nu_1} T^{[2]}_{\nu_1p_2\nu_2} &\cdots T^{[L-1]}_{\nu_{L-2}p_{L-1}\nu_{L-1}} T^{[L]}_{\nu_{L-1}p_L} \\ & \cdot \ket{p_1,\dots,p_L},
\end{split}
\end{equation}
where $T^{[i]}\in\C^{\chi_{i-1}\times d\times\chi_i}$ is a site tensor, $p_i\in\{1,\ldots,d\}$ is a physical index, and $\nu_i\in\{1,\ldots,\chi_i\}$ is a virtual index.  The open-boundary dimensions are $\chi_0=\chi_L=1$, and $\chi_i$ is the dimension of bond $(i,i+1)$.  The MPS is shown in \cref{fig:mps_mpo}(a).
An MPS is in canonical form with respect to site $i$ if
\begin{align}
    \sum_{\nu_{j-1} p_j} T^{[j]}_{\nu_{j-1}p_j \nu_j}\Bar{T}^{[j]}_{\nu_{j-1}p_j \nu_j'} &= \delta_{\nu_j \nu_j'} & \, \forall \,  j < i \label{eq:lc_form_index} \\
    \sum_{p_j \nu_j}T^{[j]}_{\nu_{j-1}p_j \nu_j}\Bar{T}^{[j]}_{\nu_{j-1}'p_j \nu_j} &= \delta_{\nu_{j-1} \nu_{j-1}'} & \, \forall \,  j > i, \label{eq:rc_form_index}
\end{align}
where $\delta$ denotes the Kronecker delta.  This is also called $i$-site canonical form.  The first condition applies to $j<i$, and the second applies to $j>i$; the contractions are shown in \cref{fig:mps_mpo}(b).  Site $i$ is the orthogonality center.  We denote tensors satisfying the first condition by $Q_R^{[j]}$ and those satisfying the second by $Q_L^{[j]}$.  QR or LQ factorizations move the center without changing the physical state.  Besides improving numerical conditioning, this gauge makes the block-state embeddings surrounding the center isometric.

A matrix product operator (MPO) represents a many-body operator as
\begin{equation}\label{eq:mpo_symb}
\begin{split}
    \hat{O} = \sum_{\substack{\{\mu_i\}_{i=1}^{L-1}\\
                                 \{q_i, p_i\}_{i=1}^{L}}}
                                    H^{[1]}_{q_1p_1\mu_1} & H^{[2]}_{\mu_1q_2p_2\mu_2} \cdots H^{[L]}_{\mu_{L-1}q_Lp_L} \\ & \cdot \ket{q_1,\dots,q_L}\bra{p_1,\dots,p_L},
\end{split}
\end{equation}
The corresponding tensor network is shown in \cref{fig:mps_mpo}(c).

Several constructions convert a many-body operator to MPO form \cite{Schollwock2011,Hubig_2017,Fr_wis_2010,Cakir2025}.  Here every MPO represents a Hamiltonian.  To construct the local equations used in a canonical sweep, the global MPO contraction must be reduced to objects associated with one site.  \Cref{fig:contractions} shows this reduction from the full expectation value to left and right environments.

\begin{figure}[t]
    \centering
    \begin{tikzpicture}
        \def\separ{1.7}
        \def\physsep{1.6}
        \node[font=\bfseries] at (-0.55,2.6*\physsep) {(a)};
        \node[mps node] (T1) at (0,0) {$T^{[1]}$};
        \node[mps node] (T2) at (\separ,0) {$T^{[2]}$};
        \node[mps node] (TL1) at (2.5*\separ,0) {$T^{[L-1]}$};
        \node[mps node] (TL) at (3.5*\separ,0) {$T^{[L]}$};
        \draw (T1) -- (T2) -- (1.5*\separ,0);
        \draw[dashed] (1.5*\separ,0) -- (2*\separ,0);
        \draw (2*\separ,0) -- (TL1) -- (TL);
        \begin{scope}[shift={(0,\physsep)}]
            \node[mpo node] (O1) at (0,0) {$H^{[1]}$};
            \node[mpo node] (O2) at (\separ,0) {$H^{[2]}$};
            \node[mpo node] (OL1) at (2.5*\separ,0) {$H^{[L-1]}$};
            \node[mpo node] (OL) at (3.5*\separ,0) {$H^{[L]}$};
            \draw (O1) -- (O2) -- (1.5*\separ,0);
            \draw[dashed] (1.5*\separ,0) -- (2*\separ,0);
            \draw (2*\separ,0) -- (OL1) -- (OL);
        \end{scope}
        \begin{scope}[shift={(0,2*\physsep)}]
            \node[mps node] (Ts1) at (0,0) {$\Bar{T}^{[1]}$};
            \node[mps node] (Ts2) at (\separ,0) {$\Bar{T}^{[2]}$};
            \node[mps node] (TsL1) at (2.5*\separ,0) {$\Bar{T}^{[L-1]}$};
            \node[mps node] (TsL) at (3.5*\separ,0) {$\Bar{T}^{[L]}$};
            \draw (Ts1) -- (Ts2) -- (1.5*\separ,0);
            \draw[dashed] (1.5*\separ,0) -- (2*\separ,0);
            \draw (2*\separ,0) -- (TsL1) -- (TsL);
        \end{scope}
        \draw (T1) -- (O1) -- (Ts1);
        \draw (T2) -- (O2) -- (Ts2);
        \draw (TL1) -- (OL1) -- (TsL1);
        \draw (TL) -- (OL) -- (TsL);
    \end{tikzpicture}
    \vspace{0.7em}

    \begin{tikzpicture}
        \def\separ{2.2}
        \def\physsep{1.6}

        \node[font=\bfseries] at (-0.5,1.5*\physsep) {(b)};
        
        \node[mps node] (T1) at (0,-1*\physsep){$T^{[1]}$};
        \node[mpo node] (O1) at (0,0){$H^{[1]}$};
        \node[mps node] (Ts1) at (0,\physsep){$\Bar{T}^{[1]}$};
        \draw (T1) -- (O1) -- (Ts1);
        \node[mps node] (Tj) at (\separ,-1*\physsep){$T^{[i-1]}$};
        \node[mpo node] (Oj) at (\separ,0){$H^{[i-1]}$};
        \node[mps node] (Tsj) at (\separ,\physsep){$\Bar{T}^{[i-1]}$};
        \draw (Tj) -- (Oj) -- (Tsj);
        \draw (T1) -- (0.4*\separ,-1*\physsep);
        \draw (Tj) -- (0.6*\separ,-1*\physsep);
        \draw[dashed] (0.4*\separ,-1*\physsep) -- (0.6*\separ,-1*\physsep);
        \draw (O1) -- (0.4*\separ,0);
        \draw (Oj) -- (0.6*\separ,0);
        \draw[dashed] (0.4*\separ,0) -- (0.6*\separ,0);
        \draw (Ts1) -- (0.4*\separ,\physsep);
        \draw (Tsj) -- (0.6*\separ,\physsep);
        \draw[dashed] (0.4*\separ,\physsep) -- (0.6*\separ,\physsep);
        \draw (Tj) -- (1.5*\separ,-1*\physsep);
        \draw (Oj) -- (1.5*\separ,0);
        \draw (Tsj) -- (1.5*\separ,\physsep);

        \begin{scope}[shift={(2.25*\separ,0)}]
            \node[font=\bfseries] at (-0.5*\separ,1.5*\physsep) {(c)};
            \node[mps node] (T1) at (0,-1*\physsep){$T^{[i+1]}$};
            \node[mpo node] (O1) at (0,0){$H^{[i+1]}$};
            \node[mps node] (Ts1) at (0,\physsep){$\Bar{T}^{[i+1]}$};
            \draw (T1) -- (O1) -- (Ts1);
            \node[mps node] (Tj) at (\separ,-1*\physsep){$T^{[L]}$};
            \node[mpo node] (Oj) at (\separ,0){$H^{[L]}$};
            \node[mps node] (Tsj) at (\separ,\physsep){$\Bar{T}^{[L]}$};
            \draw (Tj) -- (Oj) -- (Tsj);
            \draw (T1) -- (0.4*\separ,-1*\physsep);
            \draw (Tj) -- (0.6*\separ,-1*\physsep);
            \draw[dashed] (0.4*\separ,-1*\physsep) -- (0.6*\separ,-1*\physsep);
            \draw (O1) -- (0.4*\separ,0);
            \draw (Oj) -- (0.6*\separ,0);
            \draw[dashed] (0.4*\separ,0) -- (0.6*\separ,0);
            \draw (Ts1) -- (0.4*\separ,\physsep);
            \draw (Tsj) -- (0.6*\separ,\physsep);
            \draw[dashed] (0.4*\separ,\physsep) -- (0.6*\separ,\physsep);
            \draw (T1) -- (-0.5*\separ,-1*\physsep);
            \draw (O1) -- (-0.5*\separ,0);
            \draw (Ts1) -- (-0.5*\separ,\physsep);
        \end{scope}
    \end{tikzpicture}
    \caption{\textbf{MPO contractions and environments.} How is the global Hamiltonian reduced to the local contractions used during a canonical sweep?  \textbf{(a)} The full network gives the Hamiltonian expectation value $\bra\psi\hat H\ket\psi$.  \textbf{(b)} and \textbf{(c)} collect all tensors to the left and right of site~$i$ into the environments $\mathcal L^{[i]}$ and $\mathcal R^{[i]}$.  Contracting these environments with the local MPO tensor produces the effective Hamiltonian used in the projected update.}
    \label{fig:contractions}
\end{figure}
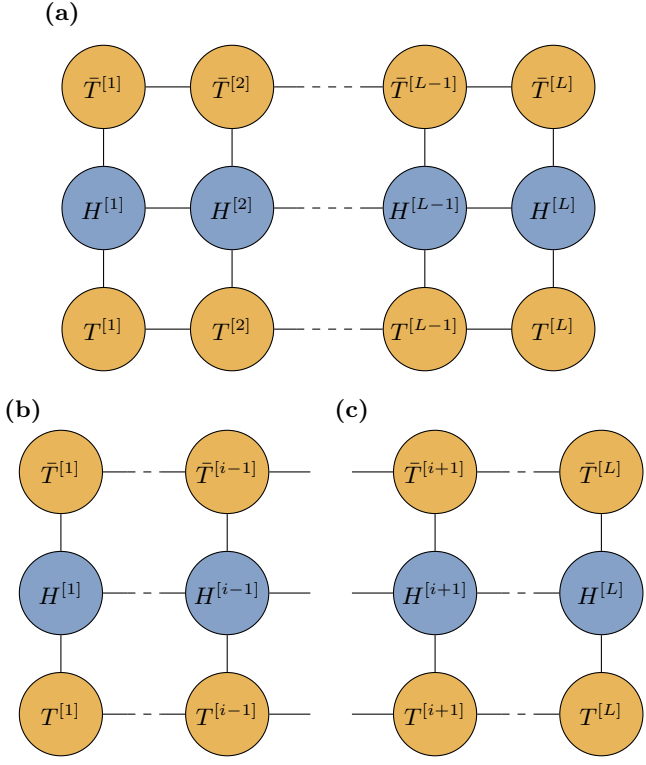

The partial contraction over sites $j<i$ defines the left environment $\mathcal L^{[i]}$, and the partial contraction over sites $j>i$ defines the right environment $\mathcal R^{[i]}$.  They are shown in \cref{fig:contractions}(b) and (c), respectively.

These canonical block embeddings and MPO environments construct the local BUG equations used in the next section.

\section{Canonical MPS/MPO realization and structural results}\label{sec:BUG}
Turning BUG into a canonical MPS algorithm requires a concrete sequence for enlarging the bond bases and expressing the coefficients in those changing bases.  We first define one uncompressed sweep toward a fixed endpoint, which serves as the Galerkin root.  A preparation pass records center tensors and MPO environments.  A reverse pass then builds updated bases and projects the recorded coefficients into them while moving back toward the root.  A final local solve updates the root coefficient tensor.  The Galerkin root stays fixed even though the orthogonality center and the tensor carrying the coefficients move.  We then derive when enlargement with the current center tensor gives the same space as retaining the previous basis explicitly, and when projection into an updated basis preserves the represented block.

\subsection{Canonical preparation and MPO environments}

We specialize the TTN construction to a chain-shaped MPS and choose the left endpoint as the root.  Each tensor in this representation carries one physical leg.  \Cref{app:non_leaf_bug} connects this representation to TTN formulations that place physical modes at the leaves.  We consider
\begin{equation}
    \partial_t \ket{\psi(t)} = -i \hat{H} \ket{\psi(t)},
\end{equation}
with a time-independent Hermitian MPO Hamiltonian $\hat H$.  The sweep requires the old MPS in site-1 canonical form,
\begin{equation}\label{eq:left_canonical_prep}
 \ket{\psi(t)}
 =C_{\mathrm{prep}}^{[1]}Q_{L,\mathrm{old}}^{[2]}
 \cdots Q_{L,\mathrm{old}}^{[L]},
\end{equation}
where every $Q_{L,\mathrm{old}}^{[j]}$ is right-isometric in the sense of \cref{eq:rc_form_index}.  We then perform one left-to-right QR pass on a working copy.  At site $i<L$, we record the current center tensor and factorize
\begin{equation}\label{eq:canonical_preparation_pass}
 C_{\mathrm{prep}}^{[i]}=Q_R^{[i]}R^{[i]},
 \qquad
 C_{\mathrm{prep}}^{[i+1]}
 =R^{[i]}Q_{L,\mathrm{old}}^{[i+1]}.
\end{equation}
After this factorization, $Q_R^{[i]}$ is installed as the site-$i$ tensor of the working preparation copy and $C_{\mathrm{prep}}^{[i+1]}$ becomes its new orthogonality-center tensor, while the incoming tensor $C_{\mathrm{prep}}^{[i]}$ is retained separately for the reverse sweep. The left isometries $Q_R^{[i]}$ generated by this same pass build the environments
\begin{align}
    \mathcal{L}^{[i]}_{\nu_{i-1}\mu_{i-1}\Bar{\nu}_{i-1}} = \sum_{\substack{p_{i-1}q_{i-1}\\
    \nu_{i-2}\mu_{i-2}\Bar{\nu}_{i-2}}} & \mathcal{L}^{[i-1]}_{\nu_{i-2}\mu_{i-2}\Bar{\nu}_{i-2}} \Bar{Q}^{[i-1]}_{R,\Bar{\nu}_{i-2}q_{i-1}\Bar{\nu}_{i-1}} \nonumber \\
    \cdot & H^{[i-1]}_{\mu_{i-2}q_{i-1}p_{i-1}\mu_{i-1}} \nonumber \\
    \cdot & Q^{[i-1]}_{R,\nu_{i-2}p_{i-1}\nu_{i-1}},
\end{align}
for $i=2,\ldots,L$, with $\mathcal L^{[1]}=1$.  This is the only center-moving preparation pass.  It preserves the represented wavefunction and does not overwrite the stored old tensors in \cref{eq:left_canonical_prep}.  We store the sequence $C_{\mathrm{prep}}^{[i]}$ for the reverse sweep.  The Galerkin root remains fixed at site~1 while the center of the preparation copy moves from site~1 to site~$L$.  During the reverse sweep, a separate tensor carrying the coefficients moves back toward the root.  We call it the working tensor.

\begin{figure*}
    \centering
    \begin{tikzpicture}
        \def\physsep{1}
        \def\separ{1.5}
        \def\thickness{0.1cm}
        \begin{scope}
            \draw (-1*\separ,-1.2*\physsep) rectangle (5.4*\separ,1.2*\physsep);
            \node[font=\bfseries] at (-0.75*\separ,0.95*\physsep) {(a)};
            \node[mps node] (T) at (0,0) {$C_{\mathrm{prep}}^{[i]}$};
            \node[right, scale=1.6] (Q) at (\separ-0.07,0) {};
            \node at (\separ,0) {$Q_{L,\mathrm{old}}^{[i+1]}$};
            \draw[-stealth] (2*\separ,0) -- (2.5*\separ,0)
            node[midway,above,align=center]
            {\scriptsize QR on $C_{\mathrm{prep}}^{[i]}$\\[-1mm]
            \scriptsize absorb $R^{[i]}$ right};
            \draw (-0.75*\separ,0) -- (T);
            \draw (T) -- (\separ-0.47,0);
            \draw (\separ+0.69,0) -- (1.75*\separ,0);
            \draw (T) -- (0,\physsep);
            \draw (Q) -- (\separ-0.07,\physsep);

            \begin{scope}[shift={(-0.5*\separ,0)}]
                \node[right] (Q) at (4*\separ,0) {$Q_R^{[i]}$};
                \node[mps node] (T) at (5*\separ,0)
                    {$C_{\mathrm{prep}}^{[i+1]}$};
                \draw (3.25*\separ,0) -- (Q) -- (T) -- (5.75*\separ,0);
                \draw (Q) -- (4*\separ,\physsep);
                \draw (T) -- (5*\separ,\physsep);
            \end{scope}
        \end{scope}

        \begin{scope}[shift={(6.2*\separ,0)}]
            \def\physlegdist{0.2}
            \draw (-0.8*\separ,-1.2*\physsep)
                rectangle (4*\separ,1.2*\physsep);
            \node[font=\bfseries] at (-0.5*\separ,0.95*\physsep) {(b)};
            \foreach \i in {-1,0,1}{
                \draw (\i*\physlegdist,-1*\physsep)
                    -- (\i*\physlegdist,\physsep);
            }
            \node[draw,ellipse,fill=mblue!75] (Heff) at (0,0)
                {$H_{\mathrm{eff}}^{[i]}$};
            \node at (0.7*\separ,0) {$=$};
            \filldraw[fill=mgreen!75]
                (\separ,-0.9*\physsep)
                rectangle (1.5*\separ,0.9*\physsep)
                node[pos=.5] {$\mathcal{L}^{[i]}$};
            \node[mpo node] (H) at (2.25*\separ,0) {$H^{[i]}$};
            \filldraw[fill=mgreen!75]
                (3*\separ,-0.9*\physsep)
                rectangle (3.5*\separ,0.9*\physsep)
                node[pos=.5] {$\mathcal{R}^{[i]}$};
            \draw (1.5*\separ,0) -- (H) -- (3*\separ,0);
            \draw (2.25*\separ,-1*\physsep)
                -- (H) -- (2.25*\separ,\physsep);
            \draw (1.5*\separ,0.5*\physsep)
                to[out=0,in=-90] (2.1*\separ,\physsep);
            \draw (1.5*\separ,-0.5*\physsep)
                to[out=0,in=90] (2.1*\separ,-\physsep);
            \draw (3*\separ,0.5*\physsep)
                to[out=180,in=-90] (2.4*\separ,\physsep);
            \draw (3*\separ,-0.5*\physsep)
                to[out=180,in=90] (2.4*\separ,-\physsep);
        \end{scope}

        \draw (-1*\separ,-4.2*\physsep)
            rectangle (5.4*\separ,-1.2*\physsep);
        \begin{scope}[shift={(0.5*\separ,-3*\physsep)}]
            \node[font=\bfseries] at (-1.25*\separ,1.5*\physsep) {(c)};
            \node[ellipse,draw,fill=morange!75] (TT) at (0,0)
                {$\widehat C^{[i]}$};
            \draw (0,1*\physsep) -- (TT) -- (1.5*\separ,0);
            \draw[line width=\thickness] (-1.3*\separ,0) -- (TT);
            \draw[-stealth] (1.75*\separ,0) -- (2.25*\separ,0)
                node[midway,above] {LQ};
            \node[circle,draw,fill=mcyan!75] (R) at (3*\separ,0)
                {$L^{[i]}$};
            \draw[line width=\thickness]
                (2.5*\separ,0) -- (R) -- (4*\separ,0);
            \node[left,outer sep=0] (Q) at (4*\separ,0)
                {$\widetilde Q_L^{[i]}$};
            \draw (4*\separ,\physsep) -- (Q) -- (4.6*\separ,0);
        \end{scope}

        \draw (-1*\separ,-7.2*\physsep)
            rectangle (5.4*\separ,-4.2*\physsep);
        \begin{scope}[shift={(0,-5.7*\physsep)}]
            \node[font=\bfseries] at (-0.75*\separ,1.15*\physsep) {(e)};
            \node[mps node] (T) at (0,0) {$C_{\mathrm{prep}}^{[i-1]}$};
            \node[basis change] (M) at (1.5*\separ,0) {$M^{[i]}$};
            \draw (-0.75*\separ,0) -- (T) -- (M);
            \draw (T) -- (0,\physsep);
            \draw[line width=\thickness] (M) -- (2.5*\separ,0);
            \node at (2.75*\separ,0) {$=$};
            \node[mps node] (T) at (3.75*\separ,0)
                {$C_t^{[i-1]}$};
            \draw (3*\separ,0) -- (T) -- (3.75*\separ,\physsep);
            \draw[line width=\thickness] (T) -- (4.5*\separ,0);
        \end{scope}

        \draw (5.4*\separ,-7.2*\physsep)
            rectangle (10.2*\separ,-1.2*\physsep);
        \begin{scope}[shift={(6.9*\separ,-4.2*\physsep)}]
            \node[font=\bfseries] at (-1.25*\separ,2.5*\physsep) {(d)};
            \node[basis change] (Mi) at (0,0) {$M^{[i]}$};
            \draw[line width=\thickness]
                (Mi) to[out=-90,in=0] (-\separ,-2*\physsep);
            \draw (Mi) to[out=90,in=0] (-\separ,2*\physsep);
            \node at (0.6*\separ,0) {$=$};
            \node[left, scale=1.75] (Qold) at (1.75*\separ,2*\physsep) {};
            \node at (1.75*\separ,2*\physsep)
                {$\scriptstyle Q_{L,\mathrm{old}}^{[i]}$};
            \draw[line width=\thickness]
                (1.75*\separ,-2*\physsep)
                -- (0.7*\separ,-2*\physsep);
            \node[left] (Qnew) at (1.75*\separ,-2*\physsep)
                {$\overline{\widetilde Q}_L^{[i]}$};
            \node[basis change] (Mi1) at (2.5*\separ,0)
                {$M^{[i+1]}$};
            \draw (Qnew)
                to[out=0,in=-90] (Mi1)
                to[out=90,in=0] (Qold);

            \draw (1.75*\separ-0.85,2*\physsep) -- (0.7*\separ,2*\physsep);
            \draw (Qnew) -- (Qold);
        \end{scope}
    \end{tikzpicture}

\caption{\textbf{One uncompressed canonical MPS-BUG sweep.} How is the BUG basis update realized as a sweep toward a fixed endpoint?  \textbf{(a)} A preparation pass records center tensors and environments.  \textbf{(b)} These environments define each local effective Hamiltonian.  \textbf{(c)} The retained input and local predictor enlarge the next block basis.  \textbf{(d)} The overlap between the previous and updated block bases is evaluated recursively.  \textbf{(e)} This overlap projects the recorded coefficients into the updated basis and moves the working tensor toward the root.  A final projected solve updates the root tensor.  Thick bonds denote enlarged virtual dimensions.  Compression and normalization are not part of this map.}
    \label{fig:bug_steps}
\end{figure*}
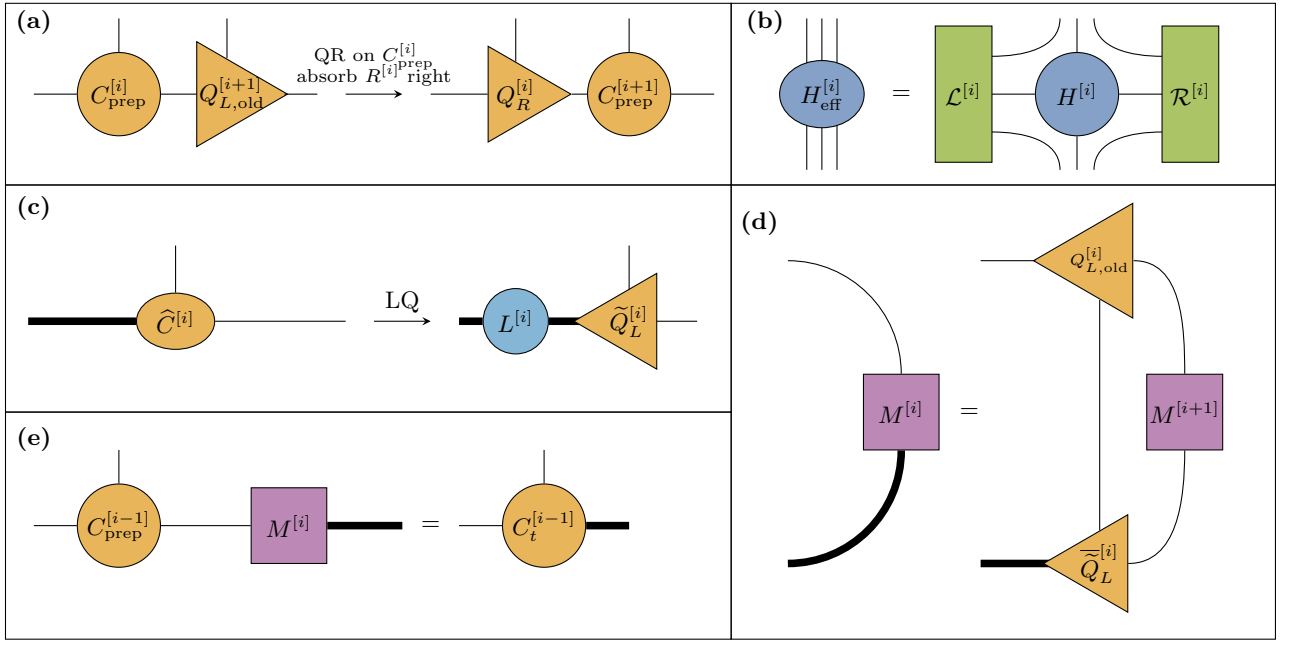
\subsection{Basis enlargement and coefficient projection}

We process the reverse sweep in the order $i=L,\ldots,2$, starting from
$C_t^{[L]}:=C_{\mathrm{prep}}^{[L]}$.  At each subsequent site, $C_t^{[i]}$ is the working tensor produced by projecting the coefficients into the basis updated at the previous site.  Its right coefficient index has been mapped into the coordinates of the new block basis on sites $i+1,\ldots,L$.  Before the local solve, the corresponding right environment is
\begin{align}
\label{eq:R_update}
    \mathcal{R}^{[i]}_{\nu'_i\mu_i\Bar{\nu}'_i} = \sum_{\substack{p_{i+1}q_{i+1}\\ \Bar{\nu}'_{i+1} \mu_{i+1}\nu'_{i+1}}} & \Bar{\widetilde Q}^{[i+1]}_{L,\Bar{\nu}'_i q_{i+1} \Bar{\nu}'_{i+1}} H^{[i+1]}_{\mu_iq_{i+1}p_{i+1}\mu_{i+1}} \nonumber \\
    \cdot & \widetilde Q^{[i+1]}_{L,\nu'_i p_{i+1} \nu'_{i+1}} \mathcal{R}^{[i+1]}_{\nu'_{i+1} \mu_{i+1} \Bar{\nu}'_{i+1}},
\end{align}
with $\mathcal R^{[L]}=1$.  The tensors on sites $i+1,\ldots,L$ form the newly constructed right-isometric block basis.  Together with $\mathcal L^{[i]}$, this defines

\begin{align}
\label{eq:effective_hamiltonian}
    H^{[i]}_{\text{eff}, \Bar{\nu}_{i-1}q_i\Bar{\nu}'_i\nu_{i-1}p_i\nu'_i} = \sum_{\mu_{i-1}\mu_i} &\mathcal{L}^{[i]}_{\nu_{i-1}\mu_{i-1}\Bar{\nu}_{i-1}} \nonumber \\
    \cdot &H^{[i]}_{\mu_{i-1}q_ip_i\mu_i}\mathcal{R}^{[i]}_{\nu'_i\mu_i \Bar{\nu}'_i}.
\end{align}

If $J_i$ denotes the isometric embedding of a site-$i$ working tensor into the full Hilbert space, then
\begin{equation}\label{eq:heff_projection}
 H_{\mathrm{eff}}^{[i]}=J_i^\dagger HJ_i.
\end{equation}
Consequently, $H_{\mathrm{eff}}^{[i]}$ is Hermitian whenever $H$ is Hermitian, and an exact local exponential preserves the Frobenius norm of the working tensor.  This local fact does not imply that a complete sweep is unitary.
\subsubsection{Local projected update and basis construction.}

Given the environments at site $i$, we evolve the working tensor locally,
\begin{align}
\label{eq:bug_local_update}
 \partial_{\tau}C^{[i]}(\tau)
 &=-iH_{\mathrm{eff}}^{[i]}C^{[i]}(\tau),
 &C^{[i]}(0)&=C_t^{[i]},\\
 \widetilde C^{[i]}&:=C^{[i]}(\delta).
\end{align}
Here $\tau\in[0,\delta]$ is a local integration variable that is reset to zero at every site, while $\delta$ is the physical duration represented by the complete endpoint sweep.  With the environments held fixed during this local solve,
\begin{equation}
 \widetilde C^{[i]}
 =\exp\!\left(-i\delta H_{\mathrm{eff}}^{[i]}\right)C_t^{[i]}.
\end{equation}
The local intervals are not accumulated across the sites: at every non-root site, $\widetilde C^{[i]}$ is used only to supply directions for the updated block basis, rather than as a successively evolved physical state.  We therefore call it the local predictor.  The final local solve at the Galerkin root supplies the coefficient update retained by the endpoint sweep.  In the alternating-endpoint composition introduced below, $\delta=h/2$, where $h$ is the complete physical time-step size. We call that updated basis the trial basis.  This is not yet the Galerkin update of the root coefficient tensor.  In the variant that enlarges the basis with the current center tensor, which we call center augmentation, the old tensor retained in the stack is
\begin{equation}\label{eq:bug_keep_tensor}
 C_{\mathrm{keep}}^{[i]}=
 \begin{cases}
  Q_{L,\mathrm{old}}^{[L]},&i=L,\\
  C_t^{[i]},&2\leq i<L.
 \end{cases}
\end{equation}
Thus the endpoint explicitly retains its old row-isometric basis tensor, whereas every internal site retains the working tensor.  We concatenate the retained tensor with the local predictor along the left virtual index,
\begin{equation}\label{eq:bug_stack}
 \widehat C^{[i]}=
 \begin{pmatrix}C_{\mathrm{keep}}^{[i]}\\ \widetilde C^{[i]}\end{pmatrix}.
\end{equation}
We matricize the stack with the left virtual index labeling the rows and the physical and right indices labeling the columns.  We then use the LQ factorization
\begin{equation}\label{eq:bug_lq}
 \widehat C^{[i]}_{(\mathrm{left})\times(\mathrm{phys,right})}
 =L^{[i]}\widetilde Q_L^{[i]}
\end{equation}
This factorization produces a row-isometric tensor
$\widetilde Q_L^{[i]}\in\C^{\chi_{i-1}'\times d\times\chi_i'}$.
Only this isometric factor is retained at site $i$.  It defines the updated right-block basis, which may be larger because it includes directions from the local predictor.

\subsubsection{Projecting coefficients into the updated block basis.}
After constructing the updated basis, we must express the recorded coefficients in its coordinates before continuing toward the root.  Let
$B_{\mathrm{old}}^{[i]}$ and $\widehat B^{[i]}$ denote the previous and updated
row-isometric right-block bases obtained by contracting
$Q_{L,\mathrm{old}}^{[i]},\ldots,Q_{L,\mathrm{old}}^{[L]}$ and
$\widetilde Q_L^{[i]},\ldots,\widetilde Q_L^{[L]}$, respectively. The objects $B_{\mathrm{old}}^{[i]}$ and $\widehat B^{[i]}$ denote the complete
contracted right-block embeddings, rather than additional site tensors, and are
not formed explicitly; \cref{fig:bug_steps}(d) evaluates their overlap
recursively from the local isometries. Their overlap is
$M^{[i]}=B_{\mathrm{old}}^{[i]}(\widehat B^{[i]})^\dagger$.  Because the
stored old MPS remains in site-1 canonical form, its tensor at every non-root site is $Q_{L,\mathrm{old}}^{[i]}$.
Starting from $M^{[L+1]}=I$, we evaluate this overlap recursively as
\begin{equation}\label{eq:bug_overlap_recursion}
 M^{[i]}_{\nu_{i-1}\nu'_{i-1}}
 =\sum_{p_i,\nu_i,\nu'_i}
 Q^{[i]}_{L,\mathrm{old},\nu_{i-1}p_i\nu_i}
 M^{[i+1]}_{\nu_i\nu'_i}
 \Bar{\widetilde Q}^{[i]}_{L,\nu'_{i-1}p_i\nu'_i}.
\end{equation}
Multiplying the recorded coefficients by this overlap projects them into the updated coordinates,
\begin{equation}\label{eq:working_center}
 C_{t,\nu_{i-2}p_{i-1}\nu'_{i-1}}^{[i-1]}
 =\sum_{\nu_{i-1}}C^{[i-1]}_{\mathrm{prep},\nu_{i-2}p_{i-1}\nu_{i-1}}
 M^{[i]}_{\nu_{i-1}\nu'_{i-1}},
\end{equation}
which supplies the working tensor for the next site.  We call this multiplication coefficient transport.  Because the map uses the block-basis overlap, we also refer to it as overlap transport.  \Cref{eq:bug_overlap_recursion,eq:working_center} specify the overlap and transport operations in canonical MPS form.  The same bookkeeping appears in TTN-BUG and recent center-rooted MPS formulations \cite{Ceruti2023,petersson2026dynamicalsimulationsschrodingersequation,bauer2026timeevolutionhybridtensor}.  After site $i$ is processed, the updated tensors on sites $i,\ldots,L$ form the right-block trial basis used by each subsequent projected equation.  No predictor on that block is integrated backward.

\subsubsection{Root Galerkin solve.}
The local solves for $i=L,\ldots,2$ supply trial bases and overlap maps.  After site~2, \cref{eq:working_center} gives $C_t^{[1]}$ in the final right-block basis.  We then solve \cref{eq:bug_local_update} once with $i=1$ and the final right environment.  We retain $\widetilde C^{[1]}$ as the root coefficient tensor and perform no further basis-enlargement step at that site.  This completes the uncompressed endpoint-rooted map $\Phi_\delta^{R\to L}$.  The map ends after the root solve and contains no rank cap, compression, or normalization.

\Cref{fig:bug_steps} summarizes how preparation, basis enlargement, coefficient projection, and the root solve fit together in one endpoint-rooted sweep.  It separates the preparation pass from the reverse basis-construction sweep and distinguishes the recorded center tensors from the previous and updated right-block bases.

The construction therefore separates the uncompressed sweep from any subsequent compression.

\subsection{Two ways to retain the previous basis}

The enlarged basis should retain the previous block space while adding directions supplied by the local predictor.  Standard rank-adaptive BUG guarantees retention by including the previous orthonormal block basis explicitly \cite{Ceruti2022B,Ceruti2023}; this inclusion is also the key hypothesis used there to show that projection into the augmented basis leaves the starting tensor unchanged.  \Cref{eq:bug_keep_tensor} does this at the endpoint.  At an internal site, center augmentation  of
Ref.~\cite{bauer2026timeevolutionhybridtensor} instead retains the working tensor, while the alternative retains the previous block basis itself.  These two stacks need not span the same space.  The following proposition is elementary linear algebra, stated here to expose the full-rank assumption implicit when the center tensor is used in place of the basis.  It gives the exact condition under which the two stacks agree before any rank cap or compression.  Suppressing the site label, let
$B_0\in\C^{r\times n}$ be a row-isometric previous block basis,
$B_0B_0^\dagger=I_r$.  Write the matricized old center as
$C_0=SB_0$, where $S\in\C^{m\times r}$ contains its coefficients in that basis.  For a local predictor
$\widetilde C\in\C^{m\times n}$, the two augmentation choices generate
\begin{equation}\label{eq:center_explicit_spaces}
 \mathcal V_{\mathrm c}
 =\operatorname{range}\!\begin{pmatrix}C_0\\ \widetilde C\end{pmatrix}^{\!\dagger},
 \qquad
 \mathcal V_{\mathrm e}
 =\operatorname{range}\!\begin{pmatrix}B_0\\ \widetilde C\end{pmatrix}^{\!\dagger}.
\end{equation}
These are the row spaces supplied by the stacked inputs, before a later bond cap or compression.

\begin{proposition}[Guaranteed retention of the previous block space]\label{prop:center_augmentation}
With the definitions above,
\begin{equation}\label{eq:center_row_condition}
 \operatorname{range}(C_0^\dagger)
 =\operatorname{range}(B_0^\dagger)
 \quad\Longleftrightarrow\quad
 \operatorname{rank}(S)=r.
\end{equation}
Consequently, $\mathcal V_{\mathrm c}=\mathcal V_{\mathrm e}$ for every $\widetilde C\in\C^{m\times n}$ if and only if $S$ has rank $r$.  For one particular predictor, rank $r$ is sufficient but need not be necessary.  The predictor itself may supply a direction absent from $C_0$.  In that case, retention of the complete previous block space is equivalent to
\begin{equation}\label{eq:center_fixed_predictor_condition}
 \operatorname{rank}\!\begin{pmatrix}C_0\\ \widetilde C\\ B_0\end{pmatrix}
 =
 \operatorname{rank}\!\begin{pmatrix}C_0\\ \widetilde C\end{pmatrix}.
\end{equation}
When $S$ is square, the condition $\operatorname{rank}(S)=r$ is equivalently nonsingularity, full column rank, or full row rank.
\end{proposition}

\begin{proof}
Since $B_0B_0^\dagger=I_r$, the map $B_0^\dagger:\C^r\to\C^n$ is an isometry and
$\operatorname{range}(C_0^\dagger)=B_0^\dagger\operatorname{range}(S^\dagger)$.
This equals $\operatorname{range}(B_0^\dagger)$ exactly when
$\operatorname{range}(S^\dagger)=\C^r$, which is equivalent to
$\operatorname{rank}(S)=r$.  This equality makes the two spaces in
\cref{eq:center_explicit_spaces} identical for every predictor.  Conversely, if $\operatorname{rank}(S)<r$, choose
$\widetilde C=C_0$, as occurs for a vanishing effective generator.  The center stack then spans only
$\operatorname{range}(C_0^\dagger)$ and does not contain the complete previous block space.  Finally, \cref{eq:center_fixed_predictor_condition} is exactly the condition that appending the rows of $B_0$ does not enlarge the center-stack row space.
\end{proof}

A two-dimensional example makes the failed guarantee explicit.  Take
\begin{equation}\label{eq:center_rank_counterexample}
 B_0=I_2,
 \qquad
 S=\begin{pmatrix}1&0\\0&0\end{pmatrix},
 \qquad
 \widetilde C=C_0=S.
\end{equation}
The center stack spans only $(1,0)$, whereas the explicit stack spans all of $\C^2$.  The previous-basis direction $(0,1)$ is missing from $\mathcal V_{\mathrm c}$.  This is a loss of a guarantee, not an unavoidable outcome for every predictor.  For example, a predictor containing $(0,1)$ restores the missing direction despite the singular $S$.

The distinction between exact and numerical rank is also essential.  Because
$C_0C_0^\dagger=SS^\dagger$, $C_0$ and $S$ have the same nonzero singular values.  Thus $S_\varepsilon=\operatorname{diag}(1,\varepsilon)$ retains the complete previous block space for every $\varepsilon\ne0$ in exact arithmetic.  A numerical-rank threshold $\tau$ may treat this direction as unresolved when $|\varepsilon|\lesssim\tau$ for an absolute criterion, or when $|\varepsilon|\lesssim\tau\sigma_{\max}(S_\varepsilon)$ for a relative criterion.  Smallness alone is not exact rank deficiency.  If a full LQ factorization pads a rank-deficient stack with arbitrary orthogonal completion rows, those rows likewise do not constitute a guaranteed copy of the missing previous-basis direction.

\Cref{prop:center_augmentation} identifies when the two enlargement choices produce the same local trial space before compression.  If its rank condition fails, center augmentation is a different construction, although the predictor may still supply the missing directions.  Explicit retention therefore provides a stronger structural guarantee, not a demonstrated accuracy advantage.  Extending the result through a sweep also requires understanding what happens when coefficients are projected between the previous and updated block bases.

\subsection{Projection between previous and updated block bases}
We now ask what can be lost in this projection.  \Cref{eq:bug_overlap_recursion} gives the overlap of two row-isometric block bases.  Earlier BUG conservation arguments use the zero-loss case, ensured by explicitly including the old space in the augmented basis \cite{Ceruti2022B,Ceruti2023}.  The following proposition records the standard Pythagorean identity for orthogonal projection in the notation needed here and extends the discussion to an arbitrary coefficient tensor and to incomplete inclusion.  Its algorithmic instance is $M^{[i]}$.

\begin{proposition}[Contractivity of overlap transport]\label{prop:overlap_transport}
Let $B_0\in\C^{r_0\times n}$ and
$\widehat B\in\C^{r_1\times n}$ be row-isometric, and define
$O=B_0\widehat B^\dagger$ and
$P_{\widehat B}=\widehat B^\dagger\widehat B$.  For every matricized coefficient tensor $X\in\C^{k\times r_0}$,
\begin{equation}\label{eq:transport_pythagoras}
 \lVert XO\rVert_F^2
 =\lVert X\rVert_F^2
 -\lVert XB_0(I-P_{\widehat B})\rVert_F^2
 \leq\lVert X\rVert_F^2.
\end{equation}
Equality for this particular $X$ holds if and only if
$XB_0(I-P_{\widehat B})=0$.  The transport is norm preserving for every $X$ if and only if
\begin{equation}\label{eq:transport_inclusion}
 \operatorname{range}(B_0^\dagger)
 \subseteq\operatorname{range}(\widehat B^\dagger),
\end{equation}
equivalently
\begin{equation}\label{eq:overlap_isometry}
 OO^\dagger
 =B_0P_{\widehat B}B_0^\dagger=I_{r_0}.
\end{equation}
In this case $O^\dagger$ is an isometry from the previous coefficient space to the updated coefficient space and $B_0=O\widehat B$.
\end{proposition}

\begin{proof}
The matrix $P_{\widehat B}$ is the orthogonal projector onto
$\operatorname{range}(\widehat B^\dagger)$, and hence
$OO^\dagger=B_0P_{\widehat B}B_0^\dagger\preceq
B_0B_0^\dagger=I_{r_0}$.  Row isometry of $B_0$ and orthogonality of the projector give \cref{eq:transport_pythagoras}.  Equality for every $X$ is equivalent to $OO^\dagger=I_{r_0}$.  Moreover,
\begin{equation}
\begin{aligned}
 I_{r_0}-OO^\dagger
 &=B_0(I-P_{\widehat B})B_0^\dagger\\
 &=\bigl[(I-P_{\widehat B})B_0^\dagger\bigr]^\dagger
  \bigl[(I-P_{\widehat B})B_0^\dagger\bigr],
\end{aligned}
\end{equation}
which vanishes exactly under \cref{eq:transport_inclusion}.  That inclusion also gives
$B_0P_{\widehat B}=B_0$, and therefore $B_0=O\widehat B$.
\end{proof}

When $r_0=r_1$, norm preservation for every coefficient tensor requires equality of the two block spaces and makes $O$ unitary.  Without inclusion, the reconstructed block
$(XO)\widehat B=XB_0P_{\widehat B}$ is the orthogonal projection of $XB_0$ into the new block space.  In particular,
\begin{equation}\label{eq:overlap_projection_error}
 \lVert XB_0-(XO)\widehat B\rVert_F^2
 =\lVert X\rVert_F^2-\lVert XO\rVert_F^2.
\end{equation}
A particular coefficient tensor can nevertheless be preserved even when the complete previous block space is not contained.  The surrounding canonical tensors define an isometric embedding.  \Cref{eq:overlap_projection_error} therefore also gives the squared physical-state projection error at this cut before any later normalization.  In the implemented recursion, \cref{eq:working_center} is precisely the coefficient transport $X\mapsto XM^{[i]}$ at the cut to the left of site $i$.  The internal center-augmentation step then includes this transported working tensor in the next stack, whereas an unaugmented trial basis need not retain it.  Retention of that particular tensor is weaker than inclusion of the complete previous block space characterized in \cref{prop:center_augmentation}.

\subsubsection{Application within the canonical sweep.}
At an internal site, the recorded center can be written
$C_{\mathrm{prep}}^{[i]}=S_iQ_{L,\mathrm{old}}^{[i]}$.  Suppose that the complete previous block space on sites $i+1,\ldots,L$ is contained in the updated block space already constructed there.  \Cref{prop:overlap_transport} then gives
$M^{[i+1]}(M^{[i+1]})^\dagger=I$, so the previous local basis represented in the current right-block coordinates,
\begin{equation}\label{eq:explicit_basis_coordinate_transport}
 B_{0,\mathrm{cur}}^{[i]}
 =Q_{L,\mathrm{old}}^{[i]}(I_d\otimes M^{[i+1]}),
\end{equation}
is row-isometric and
$C_t^{[i]}=S_iB_{0,\mathrm{cur}}^{[i]}$.  \Cref{prop:center_augmentation} therefore applies to the implemented internal stack under this inductive hypothesis.  The explicit endpoint retention in \cref{eq:bug_keep_tensor} supplies the base case.  If inclusion has already failed at a deeper cut, a later center stack cannot by itself restore the missing full-block guarantee.

Together, the endpoint base case, overlap inclusion, and full-column-rank condition at every internal site give a predictor-independent guarantee that the center-augmented sweep reproduces the explicit previous-basis trial spaces before compression.

\subsubsection{Explicit previous-basis recursion.}
An alternative that retains the previous block basis at the augmentation stage is to stack $B_{0,\mathrm{cur}}^{[i]}$ itself with the predictor, rather than retaining only $C_t^{[i]}$.  If the resulting row-isometric tensor $\widehat Q_{L,\mathrm e}^{[i]}$ spans that stack, then
\begin{equation}\label{eq:explicit_old_basis_inclusion}
 \operatorname{range}\!\left((B_{0,\mathrm{cur}}^{[i]})^\dagger\right)
 \subseteq
 \operatorname{range}\!\left((\widehat Q_{L,\mathrm e}^{[i]})^\dagger\right)
\end{equation}
by construction.  Together with the deeper-block hypothesis, this establishes inclusion of the complete previous right-block basis at the next cut and closes the induction toward the root.  At the augmentation stage, the updated block space therefore contains the previous block space.  A subsequent rank cap or compression can remove included directions, so the recursion does not guarantee this inclusion afterward.  It is not a full-sweep convergence theorem.  Failure of the hypothesis for center augmentation alone implies neither instability nor loss of a conserved quantity.  Existing BUG results that require the hypothesis therefore do not apply directly.

\subsection{Compression map and structural scope}

Compression is external to the uncompressed sweep map $\Phi_\delta$.  Let $\mathcal P_{\epsilon,\chi,r_{\min}}^{L\to R}$ denote a left-to-right MPS compression followed by normalization, where $\chi$ is the bond cap and $r_{\min}$ is the minimum retained rank permitted by the local matrix dimensions.  At each bond, let $s_k$ denote the singular values.  For $\epsilon>0$, let $r_\epsilon$ be the smallest retained rank that satisfies the relative discarded-weight target.  The implemented rank at that bond is
\begin{equation}\label{eq:relative_truncation}
 \begin{split}
 r_\epsilon&=\min\left\{r\geq1:
 \frac{\sum_{k>r}s_k^2}{\sum_k s_k^2}\leq\epsilon\right\},\\
 r&=\min\!\left(\chi,\max(r_\epsilon,r_{\min})\right).
 \end{split}
\end{equation}
If $r_\epsilon>\chi$, the hard cap overrides the tolerance.  All adaptive calculations below use $r_{\min}=2$.  This rank floor leaves at least two directions available on every nontrivial bond for subsequent adaptive growth; it is a numerical safeguard, not a restriction on the physical validity of rank-one states.  The bondwise criterion is not a global state-error bound.  The compression map can remove new basis directions supplied by either enlargement choice.  \Cref{prop:center_augmentation,prop:overlap_transport} therefore apply only to the precompression trial spaces and overlap maps inside each $\Phi_\delta$, under their stated hypotheses.

The endpoint-aware implementation reuses the canonical form returned by the uncompressed sweep.  A left-rooted sweep exits with its center at the left endpoint; the following compression traverses the chain once and leaves the center at the right endpoint.  Normalization rescales only that center tensor.  Reflection then maps the center directly to the left endpoint of the reflected network, avoiding an additional canonicalization pass.

\subsection{Alternating-endpoint schedule}\label{sec:alternating_endpoint_schedule}

For a physical step size $h>0$, each endpoint sweep has duration $\delta=h/2$.  During $\Phi_{h/2}^{R\to L}$, the Galerkin root remains fixed at the left endpoint and the working tensor moves toward it.  We obtain $\Phi_{h/2}^{L\to R}$ by reversing the MPS and MPO consistently, canonicalizing at the reflected left endpoint, applying the same uncompressed endpoint kernel, and undoing that relabeling.  Its Galerkin root is the physical right endpoint and remains fixed during the sweep.

Each $\Phi$ acts on a canonical MPS representation and returns another MPS representation.  Canonicalization leaves the represented wavefunction unchanged, and the temporary site reversal is a coordinate relabeling that is undone.  Neither operation adds physical evolution.  Neither $\Phi$ includes a rank cap, compression, or normalization.  With $r_{\min}=2$, one physical update is
\begin{equation}\label{eq:alternating_endpoint_update}
 \Psi_h=\mathcal P_{\epsilon,\chi,2}^{R\to L}\circ
 \Phi_{h/2}^{L\to R}\circ
 \mathcal P_{\epsilon,\chi,2}^{L\to R}\circ
 \Phi_{h/2}^{R\to L}.
\end{equation}
The first compression supplies the canonical entry gauge and retained rank profile for the reflected half-sweep; the second restores the public center-at-exit convention after reflection is undone.

The implementation treats $\epsilon=0$ as a singular no-tolerance-discard control.  If a hard cap is supplied, it remains active.  When the cap is inactive, this setting can retain nullspace padding and should not be interpreted as the smooth limit of a positive tolerance.

The Galerkin root switches only between the two sweeps.  Endpoint exchange specifies their order and does not introduce a separate integration principle.  The complete compressed composition in \cref{eq:alternating_endpoint_update} is not one step of standard TTN-BUG.  Spatial reflection has not been proved to produce the numerical adjoint of the endpoint sweep.  Compression can discard components and is therefore not globally reversible.  We consequently claim no second-order, symmetric, reversible, or adjoint-composition property for \cref{eq:alternating_endpoint_update}.

For the uncompressed map, a first-order result can nevertheless be transferred from the published TTN-BUG analysis under explicit hypotheses.  The precise statement and its reduction to an endpoint-rooted MPS are given in \cref{app:first_order_transfer}.  This is a specialization of the robust TTN-BUG error bound in Ref.~\cite{Ceruti2023}, not a new convergence analysis.  In particular, it does not apply automatically when center augmentation fails to retain a previous block space, and it gives neither symmetry nor reversibility.

These definitions separate what the uncompressed sweep preserves from what a later compression may discard.  The numerical tests examine both parts in turn.

\begin{table}
    \caption{\textbf{Uncompressed dense-reference refinement.} Does the complete uncompressed implementation approach independent dense dynamics as the timestep is reduced?  At $T=0.4$, the phase-aligned error decreases by approximately a factor of two at every halving of the full timestep $h$, consistently with the first-order time-discretization term in \cref{thm:mps_bug_error}.  The two basis-enlargement choices agree to the displayed precision.  The two-sweep columns omit compression $\mathcal P_{\epsilon,\chi}$, and the local Krylov tolerance is $10^{-12}$.  The short sequence is not a proof of the theorem or its hypotheses.}
    \label{tab:dense_refinement}
    \begin{ruledtabular}
    \begin{tabular}{cccc}
    $h$ & One sweep & Two, center & Two, previous basis \\
    \hline
    $0.10000$ & $1.353\times10^{-1}$ & $6.717\times10^{-2}$ & $6.717\times10^{-2}$ \\
    $0.05000$ & $6.729\times10^{-2}$ & $3.345\times10^{-2}$ & $3.345\times10^{-2}$ \\
    $0.02500$ & $3.350\times10^{-2}$ & $1.668\times10^{-2}$ & $1.668\times10^{-2}$ \\
    $0.01250$ & $1.670\times10^{-2}$ & $8.333\times10^{-3}$ & $8.333\times10^{-3}$ \\
    $0.00625$ & $8.335\times10^{-3}$ & $4.165\times10^{-3}$ & $4.165\times10^{-3}$
    \end{tabular}
    \end{ruledtabular}
\end{table}

\section{Numerical results} \label{sec:numerical_results}

We evaluate the construction in two stages.  We first compare the complete uncompressed implementation with independent dense dynamics. Then we compare the rank-adaptive BUG implementation with 2-TDVP for two 16-site models.

For normalized or unnormalized states, our primary accuracy measure is the phase-aligned state error
\begin{equation}\label{eq:phase_aligned_error}
 e_\psi=\sqrt{\max\!\left(0,2-2
 \frac{|\braket{\psi_{\mathrm{ref}}|\psi}|}
 {\lVert\psi_{\mathrm{ref}}\rVert\,\lVert\psi\rVert}\right)}.
\end{equation}
This is the distance between normalized states after minimizing over their relative global phase.  For the rank-adaptive benchmarks we report the infidelity
\begin{equation}\label{eq:infidelity}
 \mathcal I=1-
 \frac{|\braket{\psi_{\mathrm{ref}}|\psi}|^2}
 {\lVert\psi_{\mathrm{ref}}\rVert^2\lVert\psi\rVert^2}.
\end{equation}
The timestep $h$ always denotes a complete physical step.  Each sweep in the two-sweep BUG update therefore has duration $h/2$.

The 16-site calculations ran on an eight-core Apple M1 Pro with 16~GB of memory.  Numeric kernels could use all eight cores, while benchmark configurations were timed sequentially.  The implemented method is available within MQT-YAQS~\cite{YAQS,wille_mqt2024}.

\begin{table*}[t]
\caption{\textbf{Parameter comparison at $L=16$.} Both methods use relative discarded weight $\epsilon=10^{-12}$, $r_{\min}=2$, $\chi_{\max}=512$, and the same matrix-free Lanczos implementation.  Runtimes are medians of three warmed runs.  The ratio is $t_{\mathrm{2TDVP}}/t_{\mathrm{BUG}}$; values above one favor BUG.  The cap is inactive in every row.}
\label{tab:l16_matched}
\begin{ruledtabular}
\begin{tabular}{lcccccccc}
Model & $h$ & $t_{\mathrm B}$ (s) & $t_{\mathrm T}$ (s) & $t_{\mathrm T}/t_{\mathrm B}$ & $\mathcal I_{\mathrm B}$ & $\mathcal I_{\mathrm T}$ & $\chi_{\mathrm B}$ & $\chi_{\mathrm T}$ \\
\hline
TFIM & $0.01000$ & 1.598 & 2.286 & 1.431 & $6.42\times10^{-8}$ & $6.49\times10^{-9}$ & 8 & 8 \\
TFIM & $0.00500$ & 3.027 & 4.153 & 1.372 & $1.16\times10^{-7}$ & $3.20\times10^{-8}$ & 8 & 8 \\
TFIM & $0.00250$ & 5.479 & 7.706 & 1.407 & $5.25\times10^{-7}$ & $1.52\times10^{-7}$ & 8 & 8 \\
TFIM & $0.00125$ & 10.508 & 14.214 & 1.353 & $2.37\times10^{-6}$ & $7.37\times10^{-7}$ & 8 & 8 \\
Haldane-Shastry & $0.01000$ & 58.889 & 63.991 & 1.087 & $1.60\times10^{-5}$ & $1.68\times10^{-5}$ & 143 & 143 \\
Haldane-Shastry & $0.00500$ & 82.252 & 89.370 & 1.087 & $4.38\times10^{-6}$ & $4.42\times10^{-6}$ & 122 & 123 \\
Haldane-Shastry & $0.00250$ & 100.732 & 116.167 & 1.153 & $3.02\times10^{-6}$ & $2.18\times10^{-6}$ & 92 & 103 \\
Haldane-Shastry & $0.00125$ & 130.223 & 151.073 & 1.160 & $1.14\times10^{-5}$ & $5.76\times10^{-6}$ & 70 & 81
\end{tabular}
\end{ruledtabular}
\end{table*}
\subsection{Uncompressed dense-reference check}

The second stage separates two questions.  First, does the complete uncompressed implementation follow the intended site ordering, reflection, and endpoint restoration?  Second, once those checks pass, does its error decrease under timestep refinement against independent dense dynamics?  To test both questions, we use an asymmetric six-site problem with Hamiltonian
\begin{equation}\label{eq:dense_fixture}
 H=\sum_{i=0}^{4}\sum_{\alpha=x,y,z}J_i^\alpha\sigma_i^\alpha\sigma_{i+1}^\alpha
 +\sum_{i=0}^{5}\sum_{\alpha=x,y,z}h_i^\alpha\sigma_i^\alpha,
\end{equation}
with the coefficients listed in \cref{app:dense_coefficients}.  We initialize the site-index string $010011$ with site~0 as the least-significant bit and evolve to $T=0.4$.  The MPO, BUG sweeps, and final MPS are left uncompressed, with no bond cap or normalization and local Krylov tolerance $10^{-12}$.

The dense Hamiltonian is assembled independently from its Pauli terms in the Kronecker order $O_5\otimes\cdots\otimes O_0$.  Its relative difference from the MPO matrix is $9.97\times10^{-17}$, and the independently reflected matrix agrees with the reflected MPO to $1.01\times10^{-16}$.  The test problem has both an ordering gap and a reflection-asymmetry residual of $0.434$, so either mistake would be visible.  Eigensystem propagation supplies the reference state, whose norm error is $2.22\times10^{-16}$.  The implementation also passes checks of ordering, reflection, endpoint restoration, finite values, and preservation of the input tensors.  These Hamiltonian, reflection, and endpoint checks test implementation correctness.

We next use \cref{tab:dense_refinement} to test timestep responsiveness after those checks pass.

Both two-sweep variants return to the declared endpoint and agree with each other to the precision of $e_\psi$.  Their errors decrease by approximately a factor of two at every timestep halving, consistently with the first-order time-discretization term transferred in \cref{app:first_order_transfer}.  This short sequence is an implementation check, not an independent verification of that theorem or all of its hypotheses.  The one-sweep and two-sweep columns use different compositions and therefore do not test whether endpoint alternation improves accuracy.

\subsection{Rank-adaptive comparison with 2-TDVP}\label{sec:adaptive_comparison}

We compare the alternating-endpoint BUG update in \cref{eq:alternating_endpoint_update} with 2-TDVP for two $L=16$ spin chains at $T=1$.  The first is the open-boundary transverse-field Ising model
\begin{equation}\label{eq:tfim}
 H_{\mathrm{TFIM}}=-\sum_{i=1}^{L-1}Z_iZ_{i+1}-1.05\sum_{i=1}^{L}X_i,
\end{equation}
initialized in $\ket{+}^{\otimes L}$.  The second is the periodic Haldane-Shastry model \cite{haldane1988exact}
\begin{equation}\label{eq:haldane_shastry}
 H_{\mathrm{HS}}=\sum_{1\leq i<j\leq L}
 \frac{(\pi/L)^2}{\sin^2[\pi(j-i)/L]}
 \mathbf S_i\!\cdot\!\mathbf S_j,
\end{equation}
initialized in the N\'{e}el product state. The physical initial states have bond dimension one. For the adaptive benchmarks, we padded the initial MPS to bond dimension four using entries of magnitude $10^{-10}$
from a fixed random seed and imposed a minimum retained rank of two. Both algorithms started from the same padded MPS. These choices seed additional bond directions and form part of the numerical protocol considered here.  We use identical direct MPO tensors, relative discarded-weight truncation with $r_{\min}=2$, the same cap and tolerance, and the same adaptive Lanczos routine with maximum dimension 25 and tolerance $10^{-12}$. 
Independently assembled sparse Hamiltonians propagated with an exact reference supply the reference states.

The direct Ising MPO has bond dimension three.  The exact finite-state-machine Haldane-Shastry MPO has maximum bond dimension 47 and agrees with the independently assembled Hamiltonian to relative Frobenius error $2.30\times10^{-16}$.  Timing excludes construction, padding, reference evolution, warm-up, diagnostics, state-vector conversion, and file output.

\Cref{tab:l16_matched} shows that BUG is faster in all eight configurations and in every one of the 24 paired timing repetitions.  Relative to 2-TDVP, BUG reduces runtime by 26.1-30.1\% for TFIM and 8.0-13.8\% for Haldane-Shastry.  The work counters explain the direction, but not the magnitude, of this result.  In the sweep schedules used here, BUG performs $2L$ local exponential actions per physical step, whereas 2-TDVP performs $4L-7$; for $L=16$, these correspond to 32 and 57 calls to the common Krylov exponential routine.  In addition, 2-TDVP uses 1.72-1.85 times as many Krylov operator applications.  The long-range MPO and larger MPS ranks increase the relative importance of non-Krylov sweep and compression work, so the wall-clock advantage is smaller there.

Matched thresholds do not imply matched accuracy because BUG performs two global compression sweeps per step while 2-TDVP truncates at every local two-site update.  In the TFIM rows, 2-TDVP has 3.2-9.9 times smaller infidelity.  In the Haldane-Shastry rows, BUG is slightly more accurate at the two larger timesteps and 2-TDVP at the two smaller ones.  The nonmonotone errors at fixed $\epsilon$ reflect the increasing number of truncations as $h$ decreases; these rows are a cost comparison, not a convergence-order study.

\begin{figure*}[t]
\centering
\includegraphics[width=\textwidth]{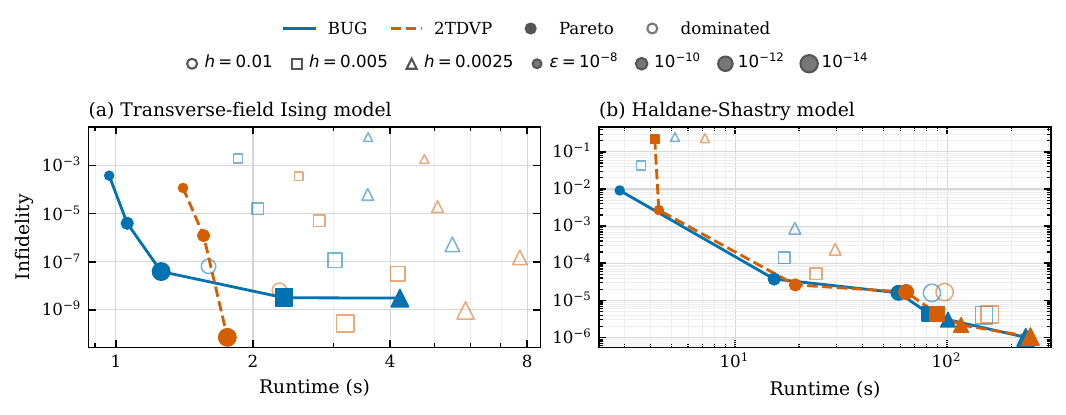}
\caption{\textbf{Runtime-accuracy trade-off at $L=16$.} Infidelity versus median runtime for \textbf{(a)} the transverse-field Ising and \textbf{(b)} Haldane-Shastry models, with $r_{\min}=2$ and $\chi_{\max}=512$.  Marker shape encodes $h$ (circle: $0.01$; square: $0.005$; triangle: $0.0025$), and marker size increases as $\epsilon$ tightens from $10^{-8}$ to $10^{-14}$.  Solid blue (BUG) and dashed orange (2-TDVP) lines connect filled Pareto points; open unconnected markers are dominated configurations.  Each runtime is the median of three repetitions.  Lines guide the eye and do not imply universal dominance.}
\label{fig:runtime_accuracy}
\end{figure*}
\begin{table}[t]
\caption{\textbf{Haldane-Shastry study.} Results at $h=0.005$, $\epsilon=10^{-12}$, and $r_{\min}=2$.  Runtimes are medians of three warmed runs.  The caps 32, 64, and 96 are attained by both methods; 512 is inactive.}
\label{tab:active_caps}
\begin{ruledtabular}
\begin{tabular}{cccccc}
$\chi_{\max}$ & $t_{\mathrm B}$ & $t_{\mathrm T}$ & $t_{\mathrm T}/t_{\mathrm B}$ & $\mathcal I_{\mathrm B}$ & $\mathcal I_{\mathrm T}$ \\
\hline
32 & 23.650 & 24.373 & 1.031 & $1.29\times10^{-4}$ & $1.68\times10^{-4}$ \\
64 & 54.320 & 55.344 & 1.019 & $5.85\times10^{-6}$ & $5.98\times10^{-6}$ \\
96 & 64.274 & 68.549 & 1.067 & $4.40\times10^{-6}$ & $4.46\times10^{-6}$ \\
512 & 82.252 & 89.370 & 1.087 & $4.38\times10^{-6}$ & $4.42\times10^{-6}$
\end{tabular}
\end{ruledtabular}
\end{table}

The trade-off view in \cref{fig:runtime_accuracy} confirms that matched-parameter speedups should not be interpreted as equal-accuracy dominance.  At $h=0.01$ and $\epsilon=10^{-14}$, TFIM gives runtime and infidelity $(1.258\,\mathrm{s},3.90\times10^{-8})$ for BUG and $(1.757\,\mathrm{s},7.19\times10^{-11})$ for 2-TDVP.  At the highest-accuracy Haldane-Shastry point, $h=0.0025$ and $\epsilon=10^{-14}$, BUG gives $(234.245\,\mathrm{s},1.01\times10^{-6})$ and 2-TDVP $(246.583\,\mathrm{s},1.05\times10^{-6})$.

The active-cap results in \cref{tab:active_caps} probe the fixed-resource regime directly.  BUG is faster and slightly more accurate at every tested cap, although the runtime difference is only 1.9-3.1\% at the two tightest caps.  The cap-96 result is already close to the cap-512 accuracy at appreciably lower cost.  We therefore interpret the table as a modest advantage in this long-range example, not as a general result for all fixed-rank problems.

\section{Discussion} \label{sec:discussion}

We formulated BUG as a sequence of canonical MPS sweeps and separated the uncompressed basis update from the compression that controls the bond dimensions. Our analysis identifies when two choices for enlarging the bond basis are equivalent. Enlarging the basis with the moving center produces the same trial space as retaining the previous basis explicitly for every predictor if and only if the coefficient map has full column rank. Transporting coefficients into an updated block basis cannot increase their Frobenius norm. This transport preserves the norm of every coefficient tensor exactly when the updated block space contains the previous block space. These conditions distinguish the two canonical MPS constructions before compression. A later compression can remove directions retained by either construction.

We compared the rank-adaptive implementation with 2-TDVP using the same initial MPS, MPO tensors, truncation parameters, bond cap, and Krylov solver. At matched parameter settings, BUG performs fewer local exponential actions and has lower median runtime in every tested case. The difference is larger for the nearest-neighbor TFIM than for the long-range Haldane--Shastry model. For the latter, the larger MPO and MPS bond dimensions increase the cost of contractions and global compression. The active-cap calculations give similar runtimes and infidelities for the two methods, so the small observed differences do not establish a clear advantage.

Matched parameters do not imply matched accuracy. At these settings, 2-TDVP is more accurate for the TFIM. For the Haldane--Shastry model, the more accurate method depends on the timestep. The runtime versus infidelity curves cross for the TFIM and remain close for the Haldane--Shastry model. Equal local truncation thresholds also do not produce equal accumulated truncation errors because the methods apply different numbers and types of singular-value decompositions. The benchmarks therefore show model-dependent trade-offs rather than a general runtime or accuracy advantage. They also use a padded initial MPS and a minimum retained rank of two. The conclusions are limited to this numerical protocol, and we have not tested their sensitivity to an exact rank-one initialization, the padding scale, or the random seed.

\Cref{app:first_order_transfer} transfers the existing first-order TTN-BUG error bound and the associated norm and energy properties to the uncompressed endpoint schedule under the stated inclusion conditions. This result does not provide an error bound for the compressed implementation used in the benchmarks. It also does not establish second-order accuracy, symmetry, reversibility, or an adjoint composition. Further tests should consider larger systems, other long-range MPO representations, different initialization choices, and comparisons at prescribed target errors. Overall, endpoint-rooted BUG provides a practical MPS time-evolution method, but its performance relative to 2-TDVP depends on the model, accuracy target, and numerical protocol.

\begin{acknowledgments}
This research was supported by the Einstein Research Unit (ERU) on quantum devices at WIAS and conducted jointly with the Technical University of Munich.  Martin Eigel acknowledges partial financial support from the German Federal Ministry of Education and Research (BMBF) under grant agreement No.~13N17160 (Q-ROM - Quantum Read-Once-Memory: Verwandlung von klassischen Daten zu Quantenzuständen). It has also been funded by 
the European Union under the Horizon Europe Programme by the European Research Council projects DA QC (Grant Agreement 101001318) and is part of the Munich Quantum Valley, which is supported by the Bavarian state government with funds from the Hightech Agenda Bayern Plus. The authors acknowledge the use of AI-based tools to support the preparation of this manuscript. All AI-assisted output was critically reviewed and verified by the authors. The authors assume responsibility for all content.
\end{acknowledgments}

\section*{Code availability}
The code supporting this work is available through MQT-YAQS~\cite{YAQS,wille_mqt2024}. 

\bibliography{references}

\appendix
\section{Physical modes at non-leaf nodes}\label{app:non_leaf_bug}
Existing TTN formulations assign physical modes to leaf nodes \cite{Ceruti2023}.  In the chain-shaped MPS of \cref{fig:mps_mpo}(a), every core instead carries a physical leg, so nonendpoint cores are internal nodes with physical modes.  This appendix explains how those legs enter the local Tucker construction.  To make this explicit, let $T$ be the tensor at site $s$ of a TTN $\mathfrak T$.  Its physical leg has dimension $d$, its leg toward the root has dimension $r$, and its $n$ child legs have dimensions $r_1,\ldots,r_n$.  The standard construction treats $T$ as the core of a Tucker tensor with multilinear rank $(r,d,r_1,\ldots,r_n)$ \cite{Ceruti2023}.  The factors associated with the root and physical modes are identity matrices of dimensions $r$ and $d$.  Each remaining factor is the subtree rooted at one child.  Because the physical mode has an identity factor, the first BUG subflow acts trivially on it when $T$ is the orthogonality center \cite{Ceruti2023}.  The effective generator is obtained by contracting the local tree tensor network operator (TTNO) tensor $H^{[s]}$ with the updated child environments and the old environment toward the root.  The compression maps in \cref{eq:alternating_endpoint_update} truncate only virtual legs and retain the physical legs exactly.  Physical modes at internal MPS sites therefore do not obstruct the specialization.

\section{Transfer of the untruncated TTN-BUG error bound}\label{app:first_order_transfer}

We record how the published robust first-order estimate for rank-adaptive TTN-BUG transfers to the uncompressed alternating-endpoint MPS schedule.  The precise rank-adaptive estimate is Theorem~5.1 of Ref.~\cite{Ceruti2023}; as explained there, its analysis combines the fixed-rank TTN proof of Ref.~\cite{ceruti2020timeintegrationtreetensor} with the matrix and Tucker rank-adaptive argument of Ref.~\cite{Ceruti2022B}.  No new general BUG convergence theorem is claimed here.

Let $\Phi_{\delta,\mathrm e}^{R\to L}$ denote the endpoint sweep of \cref{sec:alternating_endpoint_schedule} with explicit retention of the previous block basis at every internal site, exact solution of its local differential equations, and no compression.  Let $\Phi_{\delta,\mathrm c}^{R\to L}$ denote the corresponding center-augmented map.

\begin{proposition}[Reduction to endpoint-rooted TTN-BUG]\label{prop:endpoint_ttn_reduction}
After adjoining an identity factor for every physical leg as in \cref{app:non_leaf_bug}, $\Phi_{\delta,\mathrm e}^{R\to L}$ is the untruncated rank-augmenting TTN-BUG step of Ref.~\cite{Ceruti2023} on the endpoint-rooted comb tree.  If the inclusion and full-column-rank hypotheses in \cref{prop:center_augmentation,prop:overlap_transport} hold recursively at every internal cut, then
\begin{equation}
 \Phi_{\delta,\mathrm c}^{R\to L}(Y)
 =\Phi_{\delta,\mathrm e}^{R\to L}(Y)
\end{equation}
as physical states.
\end{proposition}

\begin{proof}
On an endpoint-rooted comb tree, the TTN-BUG recursion first updates the unique nontrivial child subtree and then its connection tensor, proceeding from site $L$ to the root at site~1.  Identity factors make the physical-mode subflows trivial.  At each virtual cut, the TTN construction augments the predicted basis with the previous basis, computes their overlap, and transports the old coefficients into the augmented coordinates.  These are precisely \cref{eq:bug_stack,eq:bug_overlap_recursion,eq:working_center} with explicit retention.  Its final connection-tensor equation is the root solve in \cref{eq:bug_local_update}, which proves the first statement.

Under the stated hypotheses, \cref{prop:center_augmentation} gives equality of the center-augmented and explicitly retained trial spaces at each cut, while \cref{prop:overlap_transport} makes the associated coefficient transport lossless.  Two orthonormal bases of the same trial space differ only by a unitary change of coordinates.  The projected generators are therefore unitarily similar, and exact local evolution and the overlap recursion transform covariantly under this gauge change.  Induction from the endpoint to the root then gives the same physical output.
\end{proof}

Let $A(t)$ solve $\dot A=F(t,A)$ in the full Hilbert space and let $Y^n$ be the uncompressed alternating-endpoint approximation at $t_n=nh$.  Write $\mathcal M^m$ for the MPS rank manifold used in half-step $m$ and $P_Y^m$ for the orthogonal projector onto its tangent space.  We use the hypotheses of the robust TTN-BUG estimate: in a ball containing the relevant exact and numerical states, $F$ is Lipschitz with constant $L_F$, bounded by $B_F$, and
\begin{equation}\label{eq:mps_tangent_defect}
 \left\|F(t,Y)-P_Y^mF(t,Y)\right\|\leq\eta
\end{equation}
uniformly over the half-steps.  For $F(Y)=-iHY$ in finite dimensions, the Lipschitz and local boundedness assumptions hold automatically.

\begin{theorem}[Transferred robust first-order bound]\label{thm:mps_bug_error}
Suppose the hypotheses of \cref{prop:endpoint_ttn_reduction} hold for every half-sweep, no truncation is performed, and the local differential equations are solved exactly.  If $\|Y^0-A(0)\|\leq e_0$, then for $t_n=nh\leq T$,
\begin{equation}\label{eq:mps_bug_error_bound}
 \left\|Y^n-A(t_n)\right\|
 \leq C_0e_0+C_1\eta+C_2h.
\end{equation}
The constants depend on $L_F$, $B_F$, $T$, and the chain tree, but not on small nonzero singular values of the MPS coefficient matricizations.  Thus the time-discretization contribution is first order; convergence to the full-space solution additionally requires $e_0\to0$ and $\eta\to0$.
\end{theorem}

\begin{proof}
For the left-rooted half-sweep, \cref{prop:endpoint_ttn_reduction} reduces the claim to the robust TTN-BUG estimate of Ref.~\cite{Ceruti2023}.  The stability and consistency argument underlying that estimate is the one in Refs.~\cite{ceruti2020timeintegrationtreetensor,Ceruti2022B}.  Let $\mathcal R$ be the unitary permutation that reverses the site order.  The opposite-root map is
\begin{equation}
 \Phi_{\delta}^{L\to R}
 =\mathcal R^\dagger
  \Phi_{\delta,F_{\mathcal R}}^{R\to L}\mathcal R,
 \qquad
 F_{\mathcal R}(t,Z)=\mathcal R F(t,\mathcal R^\dagger Z).
\end{equation}
Unitary conjugation preserves norms, Lipschitz and boundedness constants, tangent-space defects, and singular values.  Hence the same one-step stability and consistency estimates hold for both endpoint choices.  Applying the accumulation argument of the TTN-BUG proof to the $2n$ alternating half-steps of size $\delta=h/2$ yields \cref{eq:mps_bug_error_bound}.  This is exactly the proof of the cited TTN result with two unitarily equivalent rooted comb trees and uniform constants; no additional factor depends on small singular values.

More explicitly, let $Z^m$ be the numerical state after $m$ half-steps, so that $Z^{2n}=Y^n$, and put $e_m=\|Z^m-A(m\delta)\|$.  The stability and local-consistency estimates in the cited analyses give, after enlarging a uniform constant $c$ if necessary,
\begin{equation}
 e_{m+1}\leq (1+c\delta)e_m+c\delta(\eta+\delta).
\end{equation}
This inequality applies to either root because of the unitary conjugacy above.  Discrete Gronwall over $m=0,\ldots,2n-1$, with $2n\delta=t_n\leq T$, gives
\begin{equation}
 e_{2n}\leq C_T\bigl(e_0+\eta+\delta\bigr),
\end{equation}
and $\delta=h/2$ proves \cref{eq:mps_bug_error_bound}.  Thus the only adaptation is the alternation of two unitarily equivalent rootings; the convergence estimate itself is the published TTN-BUG result of Ref.~\cite{Ceruti2023}.
\end{proof}

For completeness, the same reduction transfers the structure-preserving statement of Ref.~\cite{Ceruti2023}.

\begin{corollary}[Uncompressed norm and energy conservation]\label{cor:uncompressed_conservation}
For $F(Y)=-iHY$ with $H=H^\dagger$, every half-sweep satisfying the inclusion hypotheses and having its root equation solved exactly preserves $\|Y\|$ and $\langle Y,HY\rangle$.  The uncompressed alternating-endpoint composition therefore preserves both quantities.
\end{corollary}

\begin{proof}
Let $J$ be the isometric embedding of the final root coefficient tensor into the augmented MPS trial space.  Inclusion makes the transported starting state exact, $JJ^\dagger Y=Y$.  The root update is
\begin{equation}
 Y^+=J\exp(-i\delta J^\dagger HJ)J^\dagger Y.
\end{equation}
The exponential is unitary, which preserves the norm.  Since it commutes with the Hermitian effective Hamiltonian $J^\dagger HJ$, it also preserves its expectation value, which equals $\langle Y,HY\rangle$ before and $\langle Y^+,HY^+\rangle$ after the update.  Applying this argument to both half-sweeps proves the result.
\end{proof}

\section{Six-site test coefficients}\label{app:dense_coefficients}
The asymmetric dense-reference problem in \cref{eq:dense_fixture} is specified by the bond vectors $\mathbf J^\alpha=(J_0^\alpha,\ldots,J_4^\alpha)$ and field vectors $\mathbf h^\alpha=(h_0^\alpha,\ldots,h_5^\alpha)$,
\begingroup
\small
\setlength{\jot}{0pt}
\begin{align*}
 \mathbf J^x &= (0.37, 0.51, 0.29, 0.63, 0.43), \\
 \mathbf J^y &= (0.31, 0.47, 0.39, 0.55, 0.35), \\
 \mathbf J^z &= (0.61, 0.33, 0.49, 0.41, 0.57), \\
 \mathbf h^x &= (0.23,-0.17, 0.31, 0.11,-0.29, 0.19), \\
 \mathbf h^y &= (0.07,-0.11, 0.13,-0.05, 0.09,-0.03), \\
\mathbf h^z &= (-0.07, 0.13, 0.05,-0.19, 0.17, 0.27).
\end{align*}
\endgroup
The five-component bond vectors mean that no bond begins at $i=5$.  Their spatial and component asymmetry makes site-ordering and reflection mistakes detectable.

\end{document}